\documentclass[conference,letterpaper]{IEEEtran}

\usepackage[utf8]{inputenc} 
\usepackage[T1]{fontenc}
\usepackage{url}
\usepackage{ifthen}
\usepackage{cite}
\usepackage[cmex10]{amsmath} 

\usepackage{amsthm}
\usepackage{booktabs}
\usepackage{enumitem,kantlipsum}

\usepackage{dirtytalk}

\usepackage{tikz}
\usepackage{tikzscale}
\usepackage{amssymb}
\usepackage{mathtools}

\usepackage[colorlinks=true,linkcolor=black,anchorcolor=black,citecolor=black,filecolor=black,menucolor=black,runcolor=black,urlcolor=black]{hyperref}

\newtheorem{theorem}{Theorem}

\newtheorem{corollary}{Corollary}

\theoremstyle{definition}
\newtheorem{example}{Example}
\newtheorem{definition}{Definition}
\newtheorem{scheme}{Scheme}

\newcommand\numberthis{\addtocounter{equation}{1}\tag{\theequation}}

\begin{document}
\title{Optimal Computational Secret Sharing} 


 \author{
   \IEEEauthorblockN{Igor L. Aureliano\IEEEauthorrefmark{1},
                     Alejandro Cohen\IEEEauthorrefmark{2}, and
                     Rafael G. L. D'Oliveira\IEEEauthorrefmark{3}}
   \IEEEauthorblockA{\IEEEauthorrefmark{1}
                    Institute of Mathematics, Statistics and Scientific Computing, University of Campinas, SP, Brazil,
                     i236779@dac.unicamp.br}
   \IEEEauthorblockA{\IEEEauthorrefmark{2}
                    Faculty of Electrical and Computer Engineering, Technion--Institute of Technology, Haifa, Israel,
                     alecohen@technion.ac.il}
   \IEEEauthorblockA{\IEEEauthorrefmark{3}
                     School of Mathematical and Statistical Sciences, Clemson University, Clemson, SC, USA,
                     rdolive@clemson.edu}
 }

\maketitle


\begin{abstract}

In $(t, n)$-threshold secret sharing, a secret $S$ is distributed among $n$ participants such that any subset of size $t$ can recover $S$, while any subset of size $t-1$ or fewer learns nothing about it. For information-theoretic secret sharing, it is known that the share size must be at least as large as the secret, i.e., $|S|$. When computational security is employed using cryptographic encryption with a secret key $K$, previous work has shown that the share size can be reduced to $\tfrac{|S|}{t} + |K|$.

In this paper, we present a construction achieving a share size of $\tfrac{|S| + |K|}{t}$. Furthermore, we prove that, under reasonable assumptions on the encryption scheme --- namely, the non-compressibility of pseudorandom encryption and the non-redundancy of the secret key --- this share size is optimal.

\end{abstract}

\section{Introduction}

Secret sharing is a cryptographic technique that distributes a secret among a group of participants such that only specific authorized subsets can reconstruct it, while unauthorized subsets gain no information. The concept was first introduced by Blakley~\cite{blakely} and Shamir~\cite{shamir}, whose schemes --- now known as $(t,n)$-threshold schemes --- allow any subset of at least $t$ participants to recover the secret, while subsets of size $t-1$ or fewer learn nothing about it in an information-theoretic sense.

In \cite{hellman2}, Karnin et al. showed that for information-theoretic $(t,n)$-threshold secret sharing, each share must be at least as large as the secret itself, i.e., of size $|S|$.\footnote{Or at least the size of the entropy of the secret $S$ if it is compressible. In this paper, we assume the secret is already compressed.} In contrast, Krawczyk \cite{krawczyk1993secret} showed that computationally secure $(t,n)$-threshold secret sharing is possible with shares smaller than the secret. Their scheme\footnote{Krawczyk \cite{krawczyk1993secret} assumes a length-preserving encryption. We adopt the same assumption in our work.}, known as Secret Sharing Made Short (SSMS), produces shares of size $\tfrac{|S|}{t} + |K|$, where $K$ is the key used in the encryption function.

In this paper, we present a computationally secure $(t,n)$-threshold secret sharing scheme, we call Pseudorandom Encryption Threshold Sharing (PETS), where each participant’s share has size $\tfrac{|S| + |K|}{t}$, with $|S|$ denoting the secret size and $|K|$ the encryption key size. Under reasonable assumptions on the encryption --- namely, the non-compressibility of our pseudorandom encryption and the non-redundancy of the key $K$ ---  we prove that this share size is optimal.  Intuitively, the collective data held by any $t$ participants must include both a ciphertext of $S$ (of size $|S|$) and the key $K$ (of size $|K|$). Distributing $|S| + |K|$ bits across $t$ participants ensures that the average per-share size cannot be smaller than $\tfrac{|S| + |K|}{t}$.

\subsection{A $(2,3)$-Threshold Example}

To show how PETS is constructed, we provide an example and compare it against both an information-theoretically secure Shamir scheme and the computational SSMS scheme proposed in \cite{krawczyk1993secret}.

\begin{figure*}
    \centering
    \includegraphics[width=0.95\linewidth]{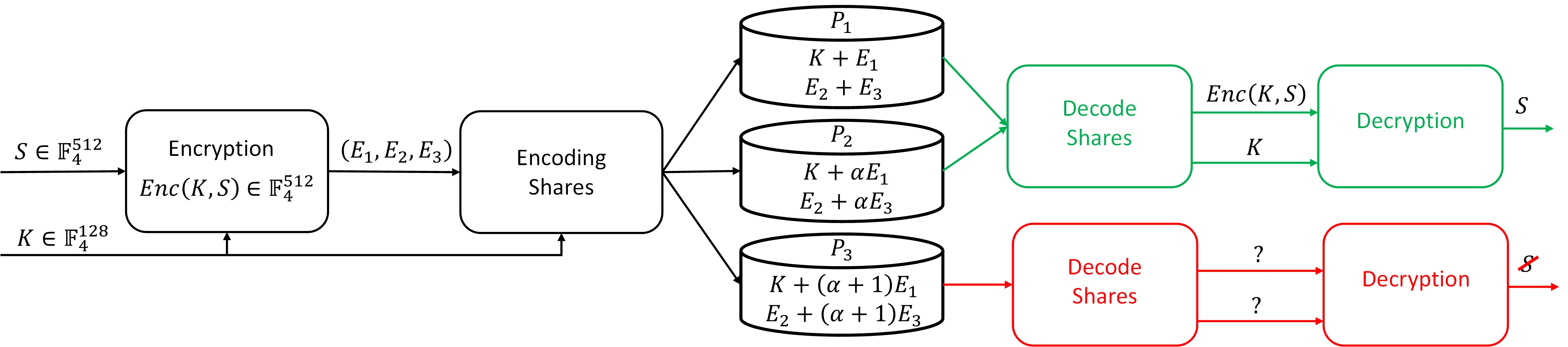}
    \vspace{-0.25cm}
    \caption{Illustration of PETS in Example~\ref{subsubsec:OurScheme23}. A secret $S \in \mathbb{F}_4^{512}$ is shared among three participants $P_1$, $P_2$, and $P_3$, such that any two participants can reconstruct $S$, while any single participant learns nothing. The secret is encrypted as $\text{Enc}(K, S)$ using a 256-bit key $K \in \mathbb{F}_4^{128}$. The ciphertext $\text{Enc}(K, S)$ is split into three parts: $E_1 \in \mathbb{F}_4^{128}$ and $E_2, E_3 \in \mathbb{F}_4^{192}$. Shares are constructed using Shamir's secret sharing scheme for $K$ with $E_1$ as the random symbol and an information dispersal algorithm for $E_2$ and $E_3$. Any two participants can combine their shares to recover $K$, reconstruct $\text{Enc}(K, S)$, and decrypt $S$, while no single participant gains any information about $K$, and therefore, about $S$. The security of the scheme relies on the pseudorandomness of $\text{Enc}(K, S)$.}
    \label{fig:example}
    \vspace{-0.4cm}
\end{figure*}

Consider a $1024$ bit secret $S$ to be shared among three participants  $\{P_1, P_2, P_3\}$. Any two participants can reconstruct $S$, while any single participant alone learns nothing. We present three schemes:
\begin{enumerate}[wide, labelwidth=0.3cm, labelindent=1pt]
\item Shamir's information-theoretic secret sharing scheme \cite{shamir},
\item Computational \say{Secret Sharing Made Short} scheme~\cite{krawczyk1993secret},
\item Our Pseudorandom Encryption Threshold Sharing scheme.
\end{enumerate}

For the computational schemes, we rely on a pseudorandom encryption 
$\text{Enc}(K, S)$ with a $256$-bit key $K$, producing a $1024$-bit ciphertext $\text{Enc}(K, S)$. We work over the finite field $\mathbb{F}_4$ with elements $\{0,1,\alpha,\alpha+1\}$ satisfying $\alpha^2 = \alpha + 1$. We identify $\mathbb{F}_2^{1024}$ with $\mathbb{F}_4^{512}$, so that, for example, $S \in \mathbb{F}_4^{512}$ and $\text{Enc}(K, S) \in \mathbb{F}_4^{512}$.

\begin{example} [Shamir Secret Sharing \cite{shamir}]
\label{subsubsec:ITscheme23}
We choose $R \in \mathbb{F}_4^{512}$ uniformly at random and define three shares: $S_1 =S+ R$, $S_2 =S+ \alpha R$, and $S_3 =S+ (\alpha+1) R$. Each participant $P_i$ receives $S_i$. 

\begin{enumerate}[wide, labelwidth=0.3cm, labelindent=1pt, label=(\roman*)]
\item \underline{Reconstruction}: Any two participants can recover $S$. For instance,  if $P_1$ and $P_2$ collaborate, they can compute
\begin{align*}
(\alpha +1) S_1 + \alpha S_2
&=
(\alpha +1) (S + R) 
+
\alpha (S + \alpha R) \\ 
&=
S, \numberthis \label{eq: reconstruct simplest example}
\end{align*}
A similar procedure works for any other pair.

\item \underline{Information-Theoretic Security}: A single share alone, say $S_1 =S+ R$, reveals nothing about $S$ because $R$ is uniform. The same argument applies to $S_2$ and $S_3$.

\item \underline{Share Size:}  Each share has $512$ symbols in $\mathbb{F}_4$. For information-theoretic security, this share size is optimal \cite{hellman2}.
\end{enumerate}
\end{example}

\begin{example} [Secret Sharing Made Short \cite{krawczyk1993secret}]
\label{subsubsec:krawczyk23}

We encrypt the secret $S \in \mathbb{F}_4^{512}$ using the key $K \in \mathbb{F}_4^{128}$ to obtain $\text{Enc}(K, S) \in \mathbb{F}_4^{512}$. We split this ciphertext into two parts, $\text{Enc}(K, S) = (E_1, E_2)$, with $E_1,E_2 \in \mathbb{F}_4^{256}$. We then choose $R \in \mathbb{F}_4^{128}$ uniformly at random and define three shares: 
\[
\begin{aligned}
S_1 &= (K + R,\; E_1 + E_2),\\
S_2 &= (K + \alpha R,\; E_1 + \alpha E_2),\\
S_3 &= \bigl(K + (\alpha + 1)R,\;\, E_1 + (\alpha + 1)E_2\bigr).
\end{aligned}
\]
Each participant $P_i$ receives $S_i$.

\begin{enumerate}[wide, labelwidth=0.3cm, labelindent=1pt, label=(\roman*)]
\item \underline{Reconstruction}: 
Any two participants can recover $K$ and $(E_1,E_2)$. For example, if $P_1$ and $P_2$ collaborate, they use the pairs $(K + R,\; K + \alpha R)$ to solve for $K$, as in \eqref{eq: reconstruct simplest example}, and $(E_1 + E_2,\; E_1 + \alpha E_2)$ to solve for $E_1,E_2$ again as in \eqref{eq: reconstruct simplest example}. Having $K$ and the complete ciphertext $\text{Enc}(K,S)$, they decrypt to obtain $S$. A similar procedure works for any pair.

\item \underline{Computational Security}:
A single share, say $S_1=(K + R,\; E_1 + E_2)$, does not reveal any information about $K$ because $R$ is uniform. The same argument applies for $S_2$ and $S_3$. Without any information about the key, the secret is computationally secured by the encryption $\text{Enc}(K,S)$.

\item \underline{Share Size}: 
Each share includes a 128-symbol block and a 256-symbol block, totaling $384$ symbols in $\mathbb{F}_4$. This is lower than what is possible in the information-theoretic approach.
\end{enumerate}
\end{example}

\begin{example}[Pseudorandom Encryption Threshold Sharing]
\label{subsubsec:OurScheme23}

As illustrated in Fig~\ref{fig:example}, we encrypt $S \in \mathbb{F}_4^{512}$ using the key $K \in \mathbb{F}_4^{128}$ to obtain $\text{Enc}(K, S) \in \mathbb{F}_4^{512}$. Unlike the previous approach, we now split the ciphertext into three parts $\text{Enc}(K, S) = (E_1, E_2, E_3)$, where $E_1 \in \mathbb{F}_4^{128}$ and $E_2, E_3 \in \mathbb{F}_4^{192}$. We define three shares:
\[
\begin{aligned}
S_1 &= \bigl(K + E_1,\;\; E_2 + E_3\bigr),\\
S_2 &= \bigl(K + \alpha\,E_1,\;\; E_2 + \alpha\, E_3\bigr),\\
S_3 &= \bigl(K + (\alpha+1)\,E_1,\;\; E_2 + (\alpha+1)\,E_3\bigr).
\end{aligned}
\]
Each participant $P_i$ receives the share $S_i$.

\begin{enumerate}[wide, labelwidth=0.3cm, labelindent=1pt, label=(\roman*)]
\item \underline{Reconstruction}: Any two participants can recover $K$ and $(E_1, E_2, E_3)$. For example, if $P_1$ and $P_2$ collaborate, they use $(K + E_1,\;\; K + \alpha\,E_1)$ to solve for $K$ and $E_1$, as in \eqref{eq: reconstruct simplest example}, and $(E_2 + E_3,\;\; E_2 + \alpha\, E_3)$ to solve for $E_2$ and $E_3$ again as in \eqref{eq: reconstruct simplest example}. Having $K$ and the complete ciphertext $\text{Enc}(K,S)$, they decrypt to obtain $S$. A similar procedure works for any pair.

\item \underline{Computational Security}:
A single share, say $S_1=(K + E_1,\; E_2 + E_3)$, does not reveal any information about $K$ because $E_1$ is pseudorandom, and therefore, $K + E_1$ is computationally indistinguishable from uniform to any computationally bounded adversary. The same argument applies for $S_2$ and $S_3$. Without any information about the key, the secret is computationally secured by the encryption $\text{Enc}(K,S)$.

\item \underline{Share Size}:
Each share includes a 128-symbol block and a 192-symbol block, totaling $320$ symbols in $\mathbb{F}_4$. We show in Theorem~\ref{thm:ThresholdOptimality} that this share size is optimal.
\end{enumerate}
\end{example}

\subsection{Related Work}

Secret sharing was originally introduced by Blakley \cite{blakely} and Shamir \cite{shamir}, who proposed information-theoretic threshold schemes. For such schemes, Karnin et al. \cite{hellman2} proved that the size of each share must be at least as large as the secret itself. 

Krawczyk \cite{krawczyk1993secret} introduced the “Secret Sharing Made Short” (SSMS) scheme which achieves computationally secure threshold secret sharing with share sizes smaller than the information-theoretic bound. This technique combines computational encryption with an Information Dispersal Algorithm (IDA) \cite{rabin1990information,naor1995optimal,beguin1998general,1023682}. Krawczyk also explored robust computational secret sharing in the same work, addressing the challenge of reconstructing the secret correctly even in the presence of tampered or corrupted shares. However, robustness is not the focus of our work, which is instead aligned with the original SSMS setting of optimizing share sizes for threshold access structures under computational security assumptions.

Robust computational secret sharing was later revisited and formalized by Bellare and Rogaway \cite{bellare2007robust}, who proved the security of Krawczyk’s robust scheme in the random oracle model. They also introduced a refined construction, HK2, that achieves robustness under standard assumptions. Additional work has extended these ideas to general access structures, as in \cite{beguin1995general,cachin1995line,mayer1997generalized,vinod2003power}, where the goal is to support more flexible definitions of authorized subsets.

Another line of research is the AONT-RS scheme \cite{resch2011aont}, which combines an all-or-nothing transform (AONT) \cite{rivest1997all} with Reed-Solomon coding \cite{reed1960polynomial} for protecting stored data. Chen et al. \cite{chen2017revisiting} revisited AONT-RS, presenting a generalized version while pointing out that the scheme can leak information if the ciphertext size is too small relative to the security parameter and threshold. They further demonstrated that AONT-RS achieves weaker security guarantees compared to SSMS. 

Our work can be viewed within the framework of Krawczyk’s original SSMS, where the goal is to minimize the share size while ensuring computational secrecy. Unlike robust computational secret sharing, we do not address tampered shares or fault tolerance. Instead, we focus on optimizing the share size under computational security.

\subsection{Paper Structure}

The paper is structured as follows. Section~\ref{sec:Preliminaries} introduces the necessary foundations of computational security, including key definitions and concepts relevant to our security analysis. In Section~\ref{sec:preliminaris}, we define fundamental notions of secret sharing and information dispersal algorithms, which serve as the building blocks for our scheme. Section~\ref{sec:main_PETS} presents our optimal computational secret sharing scheme, detailing its construction, security proofs, and optimality analysis. Finally, Section~\ref{sec:diss} concludes the paper with a discussion of potential future research directions.

\section{Computational Security}
\label{sec:Preliminaries}

This section formalizes the concept of computational security, presenting key concepts such as negligible functions, adversaries, and computational indistinguishability \cite{GOLDWASSER1984270} that we need for our security proofs. For further details, see \cite{katz2007introduction}.

We begin with negligible functions. A negligible function grows more slowly than the inverse of any positive polynomial. 

\begin{definition}[Negligible Function]
\label{def:NegligibleFunction}
A function $\mu: \mathbb{N} \to \mathbb{R}$ is called negligible if, for every 
positive polynomial $p(\cdot)$, there exists an $N \in \mathbb{N}$ such that $\mu(\lambda) < \tfrac{1}{p(\lambda)}$ for 
all $\lambda > N$.
\end{definition}

We model an attacker as a probabilistic polynomial-time (PPT) algorithm that attempts to distinguish between two distributions. We quantify success by a value called advantage, which must remain negligible for a secure scheme.

\begin{definition}[Adversary (Distinguisher)]
\label{def:Adversary}
An adversary $\mathcal{A}$ is a PPT algorithm that receives samples from either  $\mathcal{D}_0$ or $\mathcal{D}_1$ and outputs a guess in $\{0,1\}$. Its advantage  in distinguishing these distributions is
\[
\mathrm{Adv}_{\mathcal{D}_0,\mathcal{D}_1}(\mathcal{A})
=
\Bigl|\,
\Pr[\mathcal{A}(\mathcal{D}_0) = 1]
-
\Pr[\mathcal{A}(\mathcal{D}_1) = 1]
\Bigr|.
\]
\end{definition}

Two distributions $\mathcal{D}_0$ and $\mathcal{D}_1$ are said to be 
computationally indistinguishable if no PPT adversary $\mathcal{A}$ can distinguish them with more than negligible advantage. This is central to computational security proofs, since we want a real execution of a scheme to be indistinguishable from an ideal one.

\begin{definition}[Computational Indistinguishability]
\label{def:Indistinguishability}
Let $\mathcal{D}_0(\lambda)$ and $\mathcal{D}_1(\lambda)$ be two ensembles 
of distributions, parameterized by the security parameter~$\lambda$. 
The distributions are computationally indistinguishable if, for every PPT adversary 
$\mathcal{A}$,
\[
\mathrm{Adv}_{\mathcal{D}_0,\mathcal{D}_1}(\mathcal{A}) 
\;=\;
\Bigl|\,
\Pr\bigl[\mathcal{A}(\mathcal{D}_0(\lambda)) = 1\bigr]
\;-\;
\Pr\bigl[\mathcal{A}(\mathcal{D}_1(\lambda)) = 1\bigr]
\Bigr|
\]
is \emph{negligible} in $\lambda$.
\end{definition}

A pseudorandom encryption scheme ensures that encryptions of a message are indistinguishable from random strings of the same length to any PPT adversary.

\begin{definition}[Pseudorandom Encryption Scheme]
\label{def:PseudorandomEncryption}
A pseudorandom encryption scheme is a triple of algorithms 
$(\textsf{KeyGen}, \textsf{Enc}, \textsf{Dec})$ such that:
\begin{itemize}
    \item $\textsf{KeyGen}(1^\lambda)$ takes a security parameter $\lambda$ 
    and outputs a secret key $K$.
    \item $\textsf{Enc}(K, S)$ takes $K$ and a message $S$ and outputs 
    a ciphertext $C$.
    \item $\textsf{Dec}(K, C)$ takes $K$ and a ciphertext $C$ and outputs 
    the message $S$ or a failure symbol.
\end{itemize}
For any PPT adversary $\mathcal{A}$, let 
$\mathcal{D}_0 = \{\textsf{Enc}(K,S)\}$ be the distribution of encryptions of $S$ 
and $\mathcal{D}_1 = \{\mathrm{Uniform}_{|S|}\}$ be the uniform distribution over 
strings of the same length. The scheme is secure if 
\[
\mathrm{Adv}_{\mathcal{D}_0,\mathcal{D}_1}(\mathcal{A}) 
\;=\;
\Bigl|\Pr[\mathcal{A}(\mathcal{D}_0) = 1] \;-\; 
\Pr[\mathcal{A}(\mathcal{D}_1) = 1]\Bigr|
\]
is \emph{negligible} in $\lambda$ for all PPT $\mathcal{A}$.

\end{definition}

The next definition will be necessary to prove the optimality of our scheme.

\begin{definition}[Non-Redundant Encryption Scheme]
\label{def:NonRedundantEncryptionScheme}
Let $\Pi = (\mathsf{KeyGen}, \mathsf{Enc}, \mathsf{Dec})$ be an encryption scheme 
with key space $\mathcal{K}$. We say that $\Pi$ has the \emph{non-redundant key 
property} if, for every key $k \in \mathcal{K}$ and for every bit position $i$ 
in $k$, define $k^{(i)}$ to be the key obtained by flipping the $i$-th bit of $k$. 
Then there \emph{exists} a message $m$ such that $\mathsf{Dec}\bigl(k^{(i)}, \mathsf{Enc}(k, m)\bigr) \neq m$.
\end{definition}

Intuitively, this property ensures that every bit of the key is essential for decryption. If any bit could be ignored without affecting decryption, the key size could be reduced.

\section{Single, Computational Secret Sharing and Information Dispersal Algorithms}\label{sec:preliminaris}

This section provides formal definitions of secret sharing and information dispersal algorithms. We define threshold secret sharing in both the information-theoretic and computational settings and introduce information dispersal algorithms.

Threshold secret sharing allows a single secret $S$ to be distributed among $n$ participants such that any subset of at least $t$ participants can reconstruct $S$, while any subset of $t-1$ participants or fewer learns nothing about the secret $S$. This construction ensures that sensitive information remains protected unless a sufficient number of participants collaborate to recover the secret.

\begin{definition}[Information-Theoretic Threshold Secret Sharing Scheme]
\label{def:ThresholdSecretSharing}
Let $\mathcal{P} = \{P_1, P_2, \ldots, P_n\}$ be a set of $n$ participants, and let $t$ be a threshold. An information-theoretic threshold secret sharing scheme with threshold $t$ is a pair of algorithms $(\textsf{ITShare}, \textsf{ITRecon})$ defined as follows:

\noindent
\textbf{ITShare:} $\textsf{ITShare}(S, \mathcal{P}, t) \rightarrow (S_1, S_2, \ldots, S_n)$, where:
\begin{itemize}
    \item $S$ is the secret,
    \item $\mathcal{P}$ is the set of $n$ participants,
    \item $t$ is the reconstruction threshold.
\end{itemize}
This algorithm outputs a share $S_i$ for each participant $P_i$.

\noindent
\textbf{ITRecon:} $\textsf{ITRecon}\bigl(\{S_i : P_i \in A\}, A\bigr)$ takes as input a subset of participants $A \subseteq \mathcal{P}$ and their shares $\{S_i : P_i \in A\}$, and:
\begin{itemize}
    \item If $|A| \geq t$, it reconstructs $S$.
    \item Otherwise, it outputs a failure symbol (e.g., $\bot$).
\end{itemize}

The scheme must satisfy the following two properties:
\begin{enumerate}
    \item \textbf{Decodability:} $\textsf{ITRecon}\bigl(\{S_i : P_i \in A\}, A\bigr) = S$, for every subset $A \subseteq \mathcal{P}$ with $|A| \geq t$.
    \item \textbf{Security:} $I(S; \{S_i : P_i \in U\}) = 0$, for every subset $U \subseteq \mathcal{P}$ with $|U| < t$.
\end{enumerate}
\end{definition}

Shamir’s Secret Sharing \cite{shamir} is an information-theoretic $(t, n)$-threshold scheme. Shamir’s scheme serves as a building block in the construction of our proposed PETS scheme.

\begin{scheme}[Shamir’s Secret Sharing]
\label{sch:ShamirScheme}

The scheme requires the following inputs:
\begin{itemize}
    \item A secret $S \in \mathbb{F}_q$ (for a sufficiently large finite field $\mathbb{F}_q$),
    \item A threshold $t \le n$, where $n$ is the number of participants,
    \item Distinct, nonzero elements $\alpha_1, \alpha_2, \ldots, \alpha_n \in \mathbb{F}_q$.
\end{itemize}

The scheme is executed through the following steps.
\begin{itemize}
    \item \textbf{Step 1:} 
    Construct $f(x)= S + r_1 x +\cdots+ r_{t-1}x^{t-1}$, where each coefficient $r_i$ is chosen independently and uniformly at random in $\mathbb{F}_q$.

    \item \textbf{Step 2:} Define the share $S_i = f(\alpha_i)$, for $i=1,\dots,n$.

    \item \textbf{Step 3:} Send $S_i$ to participant $P_i$ for $i=1,\dots,n$.
\end{itemize}

In order to reconstruct from a subset $A \subseteq \{P_1,\dots,P_n\}$ of size at least $t$, collects its shares $\{S_i\}_{i \in A}$, use Lagrange interpolation on $\bigl(\alpha_i, S_i\bigr)_{i \in A}$ to recover the polynomial $f(x)$, and then evaluate $f(0)$ to obtain the secret $S = f(0)$.

\end{scheme}

Computational threshold secret sharing maintains the same threshold access structure but relaxes the security requirements to a computational level. Specifically, unauthorized subsets of participants (of size fewer than $t$) cannot feasibly learn any information about the secret, assuming they are bounded by computational constraints. The security guarantee relies on the hardness of breaking a pseudorandom encryption scheme.

\begin{definition}[Computational Threshold Secret Sharing Scheme]
\label{def:CompThresholdSecretSharing}
Let $\mathcal{P} = \{P_1, P_2, \ldots, P_n\}$ be a set of $n$ participants, and let $t$ be a threshold. A \emph{computational threshold secret sharing scheme} with threshold $t$ is a pair of algorithms $(\textsf{CShare}, \textsf{CRecon})$ defined as follows:

\noindent
\textsf{CShare:} $\textsf{CShare}(S, \mathcal{P}, t) \to (S_1, S_2, \ldots, S_n)$:
\begin{itemize}
    \item $S$ is a single secret,
    \item $\mathcal{P}$ is the set of $n$ participants,
    \item $t$ is the reconstruction threshold.
\end{itemize}
This algorithm outputs a share $S_i$ for each participant $P_i$.

\noindent
\textsf{CRecon:} $\textsf{CRecon}\bigl(\{S_i : P_i \in A\}, A\bigr)$:
\begin{itemize}
    \item Takes as input a subset of participants $A \subseteq \mathcal{P}$ and their shares $\{S_i : P_i \in A\}$,
    \item If $|A| \geq t$, it reconstructs $S$.
    \item Otherwise, it outputs a failure symbol (e.g., $\bot$).
\end{itemize}

The scheme must satisfy:
\begin{enumerate}
    \item \textbf{Decodability:} For every subset $A \subseteq \mathcal{P}$ with $|A| \geq t$,
    \[
    \textsf{CRecon}\bigl(\{S_i : P_i \in A\}, A\bigr) = S.
    \]
    \item \textbf{Computational Security:} For every subset $U \subseteq \mathcal{P}$ with $|U| < t$, let $\mathcal{D}_0$ denote the distribution of shares $\{S_i : P_i \in U\}$, and let $\mathcal{D}_1$ denote a uniform distribution of the same length. Then, for any PPT adversary $\mathcal{A}$:
    \[
    \mathrm{Adv}_{\mathcal{D}_0,\mathcal{D}_1}(\mathcal{A}) 
    \;=\;
    \Bigl|\,
    \Pr\bigl[\mathcal{A}(\mathcal{D}_0) = 1\bigr]
    \;-\;
    \Pr\bigl[\mathcal{A}(\mathcal{D}_1) = 1\bigr]
    \Bigr|
    \]
    is \emph{negligible} in the security parameter $\lambda$.
\end{enumerate}
\end{definition}

Both PETS and SSMS make use of an erasure code. In \cite{krawczyk1993secret} this is referred to as an information dispersal algorithm \cite{rabin1989efficient}. It is a scheme for distributing a file $M$ among $n$ participants such that any subset of at least $t$ participants can recover $M$. A simple example being a Reed-Solomon Code \cite{reed1960polynomial}.

\begin{definition}[Information Dispersal Algorithm]
\label{def:ThresholdIDA}
Let $\mathcal{P} = \{P_1, P_2, \ldots, P_n\}$ be a set of $n$ participants, and let $t$ be a threshold. An information dispersal algorithm (IDA) is a pair of algorithms $(\textsf{IDA}, \textsf{IRA})$ defined as follows:

\noindent
\textbf{IDA:} $\textsf{IDA}(M, \mathcal{P}, t) \rightarrow (X_1, X_2, \ldots, X_n)$, where:
\begin{itemize}
    \item $M$ is the file to be stored,
    \item $\mathcal{P}$ is the set of $n$ participants,
    \item $t$ is the reconstruction threshold.
\end{itemize}
This algorithm outputs a sequence of shares $(X_1, X_2, \ldots, X_n)$, where each $X_i$ is given to participant $P_i$.

\noindent
\textbf{IRA:} $\textsf{IRA}\bigl(\{X_i : P_i \in A\}, A\bigr)$:
\begin{itemize}
    \item Takes as input a subset $A \subseteq \mathcal{P}$ of participants and their shares $\{X_i : P_i \in A\}$,
    \item If $|A| \geq t$, it reconstructs $M$.
    \item Otherwise, it outputs a failure symbol (e.g., $\bot$).
\end{itemize}

The scheme must satisfy:
\begin{itemize}
    \item \textbf{Decodability:} For every subset $A \subseteq \mathcal{P}$ with $|A| \geq t$,
    \[
    \textsf{IRA}\bigl(\{X_i : P_i \in A\}, A\bigr) = M.
    \]
\end{itemize}
\end{definition}

As in \cite{krawczyk1993secret} we assume a generic IDA for our scheme. The IDA in Example~\ref{subsubsec:OurScheme23}, for example, can be obtained from a Reed-Solomon code with the polynomial $f(x)=E_2 + E_3 x$.

\section{Optimal Computational Secret Sharing}\label{sec:main_PETS}

In this section, we present our main result: a computationally secure $(t,n)$-threshold secret sharing scheme in which each share has size exactly $\tfrac{|S| + |K|}{t}$, where $|S|$ denotes the size of the secret and $|K|$ the size of the cryptographic key. This strictly improves upon the prior scheme of Krawczyk~\cite{krawczyk1993secret}, which achieves a per-share size of $\tfrac{|S|}{t} + |K|$. Moreover, we prove that our bound is optimal under standard assumptions on the encryption system (namely, pseudorandomness and non-redundancy of the key). In the remainder of this section, we describe our new construction, PETS, and give a formal proof of its security and optimality.

\begin{scheme}[Pseudorandom Encryption Threshold Sharing]
\label{sch:GeneralThresholdScheme}

The scheme requires the following inputs:
\begin{itemize}
    \item A secret $S \in \mathbb{F}_q^{|S|}$ (for a sufficiently large finite field $\mathbb{F}_q$),
    \item A key $K \in \mathbb{F}_q^{|K|}$ for a pseudorandom encryption $\textsf{Enc}$,
    \item A threshold $t \le n$, where $n$ is the number of participants,
    \item Distinct, nonzero elements $\alpha_1, \alpha_2, \ldots, \alpha_n \in \mathbb{F}_q$.
\end{itemize}

The scheme is executed through the following steps.
\begin{itemize}
    \item \textbf{Step 1:} 
    Encrypt the secret $S$ using the key $K$ to obtain the ciphertext 
    $\textsf{Enc}(K,S) \in \mathbb{F}_q^{|S|}$.

    \item \textbf{Step 2:}
    Partition the ciphertext into $t$ blocks 
    \[
    \textsf{Enc}(K,S) = \bigl(E_1, E_2,\dots,E_{t-1}, E_t\bigr),
    \]
    where $E_1,\dots,E_{t-1} \in \mathbb{F}_q^{|K|}$ and $E_t \in \mathbb{F}_q^{|S| - (t-1)|K|}$.

    \item \textbf{Step 3:}
    Construct $f(x) = K + E_1x + \cdots + E_{t-1}x^{t-1}$.

    \item \textbf{Step 4:}
    Apply an information dispersal algorithm (IDA) to $E_t$ to obtain $\textsf{IDA}(E_t,\mathcal{P},t) = (X_1,\dots,X_n)$.

    \item \textbf{Step 5:}
    Define $S_i = \bigl(f(\alpha_i), X_i\bigr)$, for $i=1,\dots,n$.

    \item \textbf{Step 6:}
    Send $S_i$ to participant $P_i$, for $i=1,\dots,n$.
\end{itemize}

In order to reconstruct from a subset $A \subseteq \mathcal{P}$ of size at least $t$, collect its shares $\{S_i\}_{i \in A}$. First, apply the $\textsf{IRA}$ to $\{X_i\}_{i \in A}$ to recover $E_t$. Then, use polynomial interpolation on the points $\bigl(\alpha_i, f(\alpha_i)\bigr)_{i \in A}$ to recover $K$ and $E_1, \dots, E_{t-1}$. Finally, assemble the full ciphertext $\textsf{Enc}(K,S)$ from $(E_1,\dots,E_t)$ and decrypt it with $K$ to recover the secret $S$.

\end{scheme}

We now show that PETS satisfies the requirements of computational security as in Definition~\ref{def:CompThresholdSecretSharing}. Specifically, we demonstrate that any subset of fewer than $t$ participants cannot gain any computationally feasible advantage in distinguishing their shares from random values. This security guarantee relies on the pseudorandomness of the encryption.

\begin{theorem}
PETS is computationally secure (Definition~\ref{def:CompThresholdSecretSharing}).
\end{theorem}
\begin{proof}
We show that for any subset $\{i_1, \dots, i_{t-1}\}\subseteq\{1,2,\dots,n\}$, the collection $\{f(\alpha_{i_1}), \ldots, f(\alpha_{i_{t-1}})\}$ is computationally indistinguishable from uniform random.

Let $R_1, \dots, R_{t-1}$ be uniform random and define the random polynomial $g(x) = K + R_1 x + \cdots + R_{t-1} x^{t-1}$. This is a Shamir Secret Sharing scheme \cite{shamir}, and thus, for any subset $\{i_1, \dots, i_{t-1}\}\subseteq\{1,2,\dots,n\}$, the collection $\{g(\alpha_{i_1}), \ldots, g(\alpha_{i_{t-1}})\}$ is uniform random. Thus, it suffices to show that $\{f(\alpha_{i_1}), \ldots, f(\alpha_{i_{t-1}})\}$ is computationally indistinguishable from $\{g(\alpha_{i_1}), \ldots, g(\alpha_{i_{t-1}})\}$.

Recall that each $E_i$ (for $i = 1, \dots, t-1$) is a pseudorandom element, so $E_i$ is computationally indistinguishable from a truly uniform element $R_i$. Formally, there exists a negligible function $\mu_i(\lambda)$ such that, for every PPT adversary $\mathcal{A}$, it holds that $\mathrm{Adv}_{E_i,R_i}(\mathcal{A}) \le \mu_i(\lambda)$.

Since the polynomials $f$ and $g$ differ only in these $t-1$ coefficients ($E_i$ vs.\ $R_i$), their distinguishing advantage satisfies
\[
\mathrm{Adv}_{f,g}(\mathcal{A}) \leq 
\sum_{i=1}^{t-1} \mathrm{Adv}_{E_i,R_i}(\mathcal{A})
\leq
\sum_{i=1}^{t-1} \mu_i(\lambda).
\]
Because each $\mu_i(\lambda)$ is negligible, their sum remains negligible. Hence $f$ and $g$ are computationally indistinguishable, implying that $\{f(\alpha_{i_1}), \ldots, f(\alpha_{i_{t-1}})\}$
is indistinguishable from $\{g(\alpha_{i_1}), \ldots, g(\alpha_{i_{t-1}})\}$, which is uniform.

Thus, any collection of $t-1$ or less participants learns no information about the key $K$. It follows then from the security of the encryption that the secret $S$ is computationally secure.
\end{proof}

We now establish that our share sizes are optimal. Intuitively, any valid $(t,n)$-threshold scheme must store both the full key and a retrievable version of the ciphertext across $t$ shares. Since we use a non-redundant key and a length-preserving pseudorandom encryption, reducing the per-share size below the threshold $\frac{|S| + |K|}{t}$ would break reconstructability. The theorem below makes this bound precise.

\begin{theorem}
\label{thm:ThresholdOptimality}
Let $S$ be a secret of size $|S|$, $K$ be a uniformly chosen $|K|$-bit key used in a pseudorandom encryption scheme $\mathrm{Enc}(K,S)$ whose output $\mathrm{Enc}(K,S)$ also has size $|S|$. Suppose we have a polynomial-time $(t,n)$-threshold computational secret sharing scheme for $S$. Then, the average share size per participant must be at least $\tfrac{|S| + |K|}{t}$.
\end{theorem}

\begin{proof}
Consider a $(t,n)$-threshold scheme that distributes $S$ among $n$ participants $P_1,\dots,P_n$ via shares $S_1,\dots,S_n$. By definition, any subset of size $t$ can reconstruct $S$. We show that for reconstruction to be possible using pseudorandom encryption with key $K$, the \emph{combined size} of those $t$ shares must be at least $|S| + |K|$.

The non-redundant key property (see Definition~\ref{def:NonRedundantEncryptionScheme}) implies that all bits of $K$ are required for decryption. Thus, any authorized subset of participants must collectively hold the entire key $K$. Moreover, since $K$ is uniform, there is no way to compress it. Since $\mathrm{Enc}(K,S)$ is computationally indistinguishable from a random string of length $|S|$, there is no polynomial-time method to store fewer than $|S|$ bits of the ciphertext and still later decompress to recover $\mathrm{Enc}(K,S)$.

\begin{figure*}
    \centering
    \includegraphics[width=0.95\linewidth]{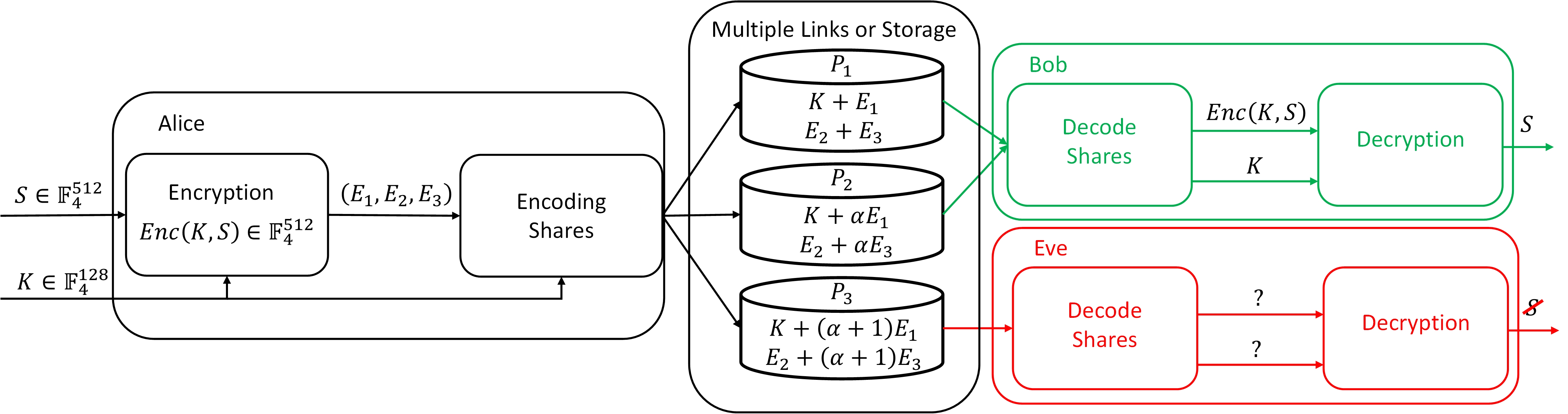}
    \vspace{-0.25cm}
    \caption{Secret sharing for secure communication over multiple links or storage. Alice is the encryption and encoding operations on the left, the three links/storeges are the three participants $P_1$, $P_2$, and $P_3$ in the middle, and Bob and Eve are in green and red, respectively, on the right side.}
    \label{fig:example_com}
    \vspace{-0.4cm}
\end{figure*}

Hence, for a subset $A\subseteq\{1,\dots,n\}$ with $|A|=t$ to reconstruct $S$, the shares $\{S_i : i\in A\}$ must collectively contain at least $|K| + |S|$ bits of information. Otherwise, decryption would fail either because $K$ is incomplete or because $\mathrm{Enc}(K,S)$ cannot be retrieved. Since there are $t$ participants in any authorized subset, and each one holds exactly one share, the average share size in that subset must be at least $\frac{|K|+|S|}{t}$. 

\end{proof}

Theorem~\ref{thm:ThresholdOptimality} provides a lower bound on the share size necessary for any computationally secure $(t, n)$-threshold secret-sharing scheme. By achieving this bound, PETS demonstrates optimality in terms of share size.

\begin{corollary}
    PETS achieves the optimal average share size.
\end{corollary}

\section{Fundamental Gains in Information Rate}
\label{sec:FundamentalGains}

Secret sharing can be applied to secure communication over multiple links \cite{ozarow1984wire,el2007wiretap,cai2010secure,el2012secure}, as illustrated in Fig.~\ref{fig:example_com}. Consider a sender, Alice, and a receiver, Bob, connected by $n$ parallel links, with an eavesdropper, Eve, capable of observing up to $t-1$ of these links. Treating each link as a \say{participant} in a $(t,n)$-threshold scheme, Alice splits her message into $n$ shares -- one share per link -- and transmits each share across its respective link. Bob, who receives all $n$ shares, can reconstruct the original message, whereas Eve, intercepting at most $t-1$ links, learns nothing about the message.  Performance is then measured by the information rate.

\begin{definition}[Information Rate]
\label{def:InformationRate}
The \emph{information rate} of a secret-sharing scheme is defined as
\[
  \mathrm{Rate}
  =
  \frac{\text{size of secret}}{\text{total size of all shares}}.
\]
\end{definition}

In this context, PETS provides a fundamental advantage when the total number of participants $n$ and the threshold $t$ grow linearly together. Specifically, if $t = \delta n$ for a fixed constant $\delta \in (0,1]$, the classical schemes of Shamir and SSMS force the information rate down to zero. In contrast, PETS achieves a constant, non-zero rate. This contrast is illustrated in the examples below.

\begin{example}[\boldmath$(n,n)$-Threshold Shamir Secret Sharing]
\label{ex:ShamirNN}
Consider Shamir’s scheme for the special case of threshold $n$. A secret $S$ of size $|S|$ is distributed into $n$ shares, each of which has size $|S|$. Hence, the total size of all shares is $n \cdot |S|$. By Definition~\ref{def:InformationRate},
\[
  \mathrm{Rate}
  =
  \frac{|S|}{n \cdot |S|}
  =
  \frac{1}{n}
  \;\longrightarrow\;
  0
  \quad (\text{as } n \to \infty).
\]
Thus, as $n$ grows large, Shamir’s rate goes to zero.
\end{example}

\begin{example}[SSMS Scheme]
\label{ex:SSMSNN}
Next, consider Krawczyk’s SSMS \cite{krawczyk1993secret} with threshold $n$. Each share has size
\(\tfrac{|S|}{n} + |K|\).
Hence, the total size of all $n$ shares is $|S| + n \cdot |K|$. Thus,
\[
  \mathrm{Rate}
  \;=\;
  \frac{|S|}{|S| + n \cdot |K|}
  \;\longrightarrow\;
  0
  \quad (\text{as } n \to \infty).
\]
While SSMS can yield shares smaller than those in purely information-theoretic approaches, the overall rate nevertheless still goes to zero with increasing $n$.
\end{example}

\begin{example}[PETS Scheme]
\label{ex:PETSNN}
In contrast, consider the PETS construction (Scheme~\ref{sch:GeneralThresholdScheme}) in the $(n,n)$-threshold scenario. Each share has size $\frac{|S| + |K|}{n}$, so the total size of all $n$ shares is $|S| + |K|$. By Definition~\ref{def:InformationRate}, the information rate is
\[
  \mathrm{Rate}
  \;=\;
  \frac{|S|}{\,|S| + |K|\,}.
\]
This does not vanish as $n \to \infty$. Hence, unlike Shamir’s and SSMS, PETS retains a constant rate as the number of communication links (participants) grows.
\end{example}

As a motivating application, recent work \cite{cohen2023absolute} has shown how secret sharing can secure wireless communications across multiple frequency channels. In this setting, each frequency channel corresponds to a link (i.e., a “participant”) in an $(n,t)$-threshold scheme, and an eavesdropper can obtain information from at most $t-1$ frequency channel measurements. To make eavesdropping more difficult, one can increase $n$ (i.e., use more frequency channels) so that there is a larger ``blind region'', i.e. regions where the eavesdropper can obtain information from at most $t-1$ measurements. However, if one relies on Shamir’s or SSMS with threshold $n$, the information rate quickly goes to zero as $n$ grows. By contrast, PETS maintains a constant nonzero rate
\[
  \mathrm{Rate}_{\text{PETS}}
  \;=\;
  \frac{|S|}{\,|S| + |K|\,},
\]
thus preserving high throughput even as the number of frequency channels grows.

In general, we obtain the following result.

\begin{theorem}[PETS Rate with $t = \delta\,n$]
\label{thm:PETSRateWithDeltaN}
Let $\delta$ be a fixed constant in $(0,1]$ and consider an $(n,t)$-threshold PETS construction where $t = \delta \, n$. If $|S|$ and $|K|$ denote the sizes of the secret and key respectively, then the information rate satisfies
\[
  \mathrm{Rate}
  =
  \delta 
  \cdot
  \frac{|S|}{\,|S| + |K|\,}.
\]
\end{theorem}
\begin{proof}
In the $(t,n)$-threshold PETS, each share is $\frac{|S| + |K|}{t}$. Since $t = \delta\,n$, each share has size $\frac{|S| + |K|}{\delta n}$.
Hence, the total size of all $n$ shares is $n 
  \cdot
  \frac{|S| + |K|}{\delta n}
  =
  \frac{|S| + |K|}{\delta}$.
By Definition~\ref{def:InformationRate}, the rate is
\[
  \mathrm{Rate}
  \;=\;
  \frac{|S|}{\,\frac{|S| + |K|}{\delta}\,}
  \;=\;
  \delta 
  \,\cdot\,
  \frac{|S|}{\,|S| + |K|\,}.
\]
\end{proof}

\section{Discussion and Future Work}\label{sec:diss}

This work revisits the computational secret-sharing paradigm introduced by Krawczyk~\cite{krawczyk1993secret}, focusing on optimizing share sizes within the original SSMS setting. Under reasonable assumptions regarding the employed encryption, PETS achieves an optimal per-share size of $\tfrac{|S| + |K|}{t}$. Furthermore, PETS attains significant gains in the information rate, especially when scaling the threshold linearly with the number of participants. Unlike traditional schemes, whose rates diminish rapidly as the number of participants increases, PETS maintains a constant, nonzero rate, making it particularly suitable for secure communications scenarios requiring many parallel links or frequency channels.

While robustness has been explored in follow-up works to SSMS, such as the robust schemes of Bellare and Rogaway \cite{bellare2007robust}, our scheme focuses solely on achieving computational security with minimal share size. Extending PETS to incorporate robustness, particularly in the presence of adversarially corrupted shares or participants, is a promising direction for future research. Such an extension could involve integrating verifiable sharing techniques \cite{4568164,4568297} to enhance resilience while preserving the efficiency of PETS. 

Additionally, our approach opens up several avenues for further exploration. One potential direction is to adapt PETS to generalized access structures \cite{ito,benaloh,stinson}, where subsets of participants are defined by arbitrary monotone structures rather than a fixed threshold. Many applications require more flexible or hierarchical access structures, where certain participants may hold higher privilege or weighted votes in reconstructing the secret. Extending the polynomial-based framework of PETS to support these configurations could achieve similar benefits of reduced per-share size. 

\section*{Acknowledgment}

RD was supported by NSF under Grant CNS-214813.

\bibliographystyle{IEEEtran}
\bibliography{ref}

\end{document}